\documentclass[]{elsarticle}
\usepackage[english]{babel}
\usepackage{verbatim}
\usepackage{amsmath,amssymb,array}
\usepackage{amsthm}
\usepackage[hidelinks]{hyperref}
\usepackage{a4wide}
\usepackage{lineno}
\usepackage[algoruled,linesnumbered,noend]{algorithm2e}
\usepackage[color=green!30]{todonotes}

\modulolinenumbers[5]

\renewcommand\leq\leqslant
\renewcommand\geq\geqslant

\newtheorem{theorem}{Theorem}
\newtheorem{lemma}[theorem]{Lemma}

\newcommand{\fullversion}[1]{#1}
\newcommand{\confversion}[1]{}

\usepackage{etoolbox}

\newcommand{\appendixproof}[2]{
#2
}

\newcommand{\Pb}[4]{%
\begin{center}
  \begin{tabular}{|l|}%
  \hline
    \begin{minipage}[c]{.95\textwidth}
      \smallskip%
      \par\noindent%
      \textsc{#1}
      \medskip%
      \par\noindent%
      $\bullet$
      \textbf{\textsf{Input}}: #2%
      \par\noindent%
      $\bullet$
      \textbf{\textsf{#4}}:
      #3%
      \smallskip%
      \par\noindent%
    \end{minipage}
  \\\hline
  \end{tabular}%
\end{center}
}%

\begin{document}

\begin{frontmatter}

\title{Parameterized Complexity and Approximation Issues for the Colorful Components Problems}

\author[rd]{Riccardo Dondi}
\ead{riccardo.dondi@unibg.it}
\address[rd]{Dipartimento di Scienze umane e sociali,  Universit\`a degli Studi di Bergamo, ITALY}

\author[fs]{Florian Sikora}
\ead{florian.sikora@dauphine.fr}
\address[fs]{Universit\'{e} Paris-Dauphine, PSL Research University, CNRS UMR 7243, LAMSADE, 75016 PARIS, FRANCE}



\begin{abstract}
The quest for colorful components 
(connected components where each color is associated with at most one vertex) 
inside a vertex-colored graph has been widely considered in the last ten years. 
Here we consider two variants, Minimum Colorful Components (MCC) and Maximum Edges in transitive Closure (MEC),
introduced in 2011 in the context of orthology gene identification in bioinformatics. 
The input of both MCC and MEC is a vertex-colored graph.
MCC asks for the removal of a subset of edges, so that the
resulting graph is partitioned in the minimum number of colorful connected components;
MEC asks for the removal of a subset of edges, so that the
resulting graph is partitioned in colorful connected components and the number of edges in the transitive closure of such a graph is maximized.
We study the parameterized and approximation complexity of MCC and MEC, for general and restricted 
instances. 

For MCC on trees we show that the problem is basically equivalent to Minimum Cut on Trees, thus MCC is not approximable within factor $1.36 - \varepsilon$, it is fixed-parameter tractable and it admits a poly-kernel (when 
the parameter is the number of colorful components).
Moreover, we show that MCC, while it is polynomial time solvable on paths, 
it is NP-hard even for graphs with constant distance to disjoint paths number.
Then we consider the  parameterized complexity of MEC when parameterized by the number $k$ of edges in the transitive closure of a solution (the graph obtained by removing edges so that it is partitioned in colorful connected components).
We give a fixed-parameter algorithm for MEC paramterized by $k$ and, when the input graph is a tree, we give a poly-kernel.
\end{abstract}


\end{frontmatter}


\section{Introduction}

The quest for colorful components 
inside a vertex colored graph has been 
a widely investigated problem in the last years, 
with application for example in 
bioinformatics~\cite{DBLP:journals/tcbb/LacroixFS06,DBLP:conf/cpm/BrucknerHKNTU12,DBLP:journals/tcs/DondiFV13}. 
Roughly speaking, given a vertex-colored graph, the problem asks to find 
the colorful components of the graph,  that is connected components
that contain at most one vertex of each color. While most of the approaches have 
focused on the identification of a single connected colorful component. 
The identification of the minimum number of
colorful connected components that match 
a given motif has been considered only 
in~\cite{DBLP:journals/tcbb/BetzlerBFKN11,DBLP:journals/jda/DondiFV11}.

Here we consider a similar framework, where instead of looking for a single colorful component
inside a vertex-colored graph, we ask for a partition of the graph vertices 
in colorful components. 
This approach has been proposed in bioinformatics, and more specifically in comparative genomics. 
In this context, a fundamental task is to 
infer the relations between genes in different genomes and, more precisely, to infer which genes are orthologous.
Genes are orthologous when they originate via a speciation event\footnote{A speciation is a evolutionary process from which a biological species evolves into two new species.} from a gene of an ancestral genome.
In 2011, Zheng et al. proposed a graph approach aiming to identify disjoint orthology sets,
where each of such sets corresponds to a colorful component in the given graph~\cite{DBLP:conf/wabi/ZhengSLS11} and the colorful components associated with the orthology sets are disjoint.

Different combinatorial problem formulations, based on different objective functions, 
have been proposed and studied in this direction~\cite{DBLP:conf/wabi/ZhengSLS11,DBLP:journals/algorithmica/AdamaszekP15}.
Here, we considered two such approaches: 
\textsc{Minimum Colorful Components} (MCC) and \textsc{Maximum Edges in transitive Closure} (MEC).
Given a vertex-colored graph, both combinatorial problems 
ask for the removal of some edges 
so that the resulting graph is partitioned in colorful components, 
but with different objective functions. 
The former aims to minimize the number of connected colorful components,
while the latter aims to maximize the transitive closure of the resulting graph.
A related but different problem has been considered in~\cite{DBLP:conf/cpm/BrucknerHKNTU12},
where the objective function is the minimization of edge removal, so that 
the computed graph consists only of colorful components.
\fullversion{Note that in the problems studied in this paper, the number of removed edges is never part of the objective function.}

\paragraph{Previous Results}
Given a graph on $n$ vertices,
MCC is known not only to be NP-hard, but also not approximable within factor 
$O(n^{1/14 - \varepsilon})$ unless P=NP~\cite{DBLP:journals/algorithmica/AdamaszekP15}.
It is easy to see that the reduction leading to this inapproximability result 
implies also that MCC 
cannot be solved in time $n^{f(k)}$ for any function $f$, where $k$ is the number of colorful components. 
In the parameterized complexity vocabulary, it means that it is not in the XP class.

MEC is known to be APX-hard even when colored by at most three colors (while it is solvable in polynomial time for two colors), 
and, unless P=NP, it is not approximable 
within factor $O(n^{1/3 - \varepsilon})$ when the number of colors is arbitrary, even when the input graph is a tree where each color appears at most twice~\cite{DBLP:conf/iwoca/AdamaszekBP14}. 
A heuristic to solve MEC is presented in~\cite{DBLP:conf/wabi/ZhengSLS11}, while in~\cite{DBLP:conf/iwoca/AdamaszekBP14}, the authors present a polynomial-time $\sqrt{2 \cdot OPT}$ approximation algorithm.

\paragraph{Contributions and organization of the paper}
In this paper we investigate deeper the complexity of MCC and MEC. 
More precisely, we show in Section~\ref{sec:MCC:trees} that MCC on trees is essentially equivalent to \textsc{Minimum MultiCut} on Trees, thus MCC is not approximable within factor $1.36 - \varepsilon$ unless P=NP for any $\varepsilon > 0$, but 2-approximable, it is fixed-parameter tractable (but not in subexponential-time) and it admits a poly-kernel (when the parameter is the number of colorful components).
Moreover, in Section~\ref{sec:MCCDisjPath} we show that MCC is easily solvable in polynomial time on paths, while  it is not in XP class when parameterized  by the structural parameter Distance to Disjoint Paths.

Then we consider the  parameterized complexity of MEC with respect to the number $k$ of edges in the transitive closure of a solution. 
For this parameter  we give in Section~\ref{sec:FPTMEC} a parameterized algorithm, by reducing the problem to an exponential kernel.
We use a similar idea in Section~\ref{sec:kernelMECtrees}, to improve it to a poly-kernel for MEC when the input graph is a tree.
Finally, we show in Section~\ref{sec:MECDisjPath} that results similar to those of Section~\ref{sec:MCCDisjPath}, hold also for MEC.
A preliminary version of this work appeared in~\cite{DBLP:conf/cie/DondiS16}.

\section{Definitions}
\label{sec:Def}

In this section we introduce some preliminary definitions.
For any positive integer $x$, $[x]$ denotes the set $\{1,2,\ldots,x-1,x\}$. 
Consider a set of colors $C=\{c_1,\dots, c_q  \}$.
A $C$--colored graph $G=(V,E,C)$ is a graph where every vertex in $V$ is associated with a color in $C$; the color associated with a vertex $v \in V$ is denoted by $c(v)$.
If $\mathcal{C}$ is a class of graphs, the \emph{distance to  $\mathcal{C}$} of a graph $G$ is the minimum number of vertices to remove from $G$ to get a graph in  $\mathcal{C}$.
A connected component induced by a vertex set $V' \subseteq V$ is called \emph{a colorful component},
if it does not contain two vertices having the same color. 
If a graph has $t$ connected components where each component $i \in [t]$ has exactly $n_i$ vertices, the number of edges in its transitive closure is defined by $\sum_{i=1}^t \frac{n_i (n_i -1)}{2}$. 
\fullversion{In other words, for each connected component, it is the maximum number of possible edges
connecting vertices of that component.}

Next, we introduce the formal definitions of the optimization problems we deal with.

\Pb{Minimum Colorful Components (MCC)}{a $C$-colored graph $G=(V,E,C)$.}{remove a set of edges $E' \subseteq E$ such that each connected component in
$G' = (V, E \setminus E', C)$ is colorful, and the number of connected components of $G'$ is minimized.}{Output}


\Pb{Maximum Edges in transitive Closure (MEC)}{a $C$-colored graph $G=(V,E,C)$.}{remove a set of edges $E' \subseteq E$ such that each connected component in
$G' = (V, E \setminus E', C)$ is colorful, 
and the number of edges in the transitive closure of $G'$ is maximum.}{Output}

The parameterized versions of MCC and MEC are defined analogously (and abusively denoted with the same names),
with the addition in the input of an integer $k$, that denotes the number of connected components in $G'$ for MCC and the number of edges in the transitive closure of $G'$ for MEC.



Notice that, when considering an instance of MCC and MEC, we assume that $E$  contains no edge $\{u,v\}$ with $c(u)=c(v)$, otherwise such an edge can be safely deleted from $E$ as $u$ and $v$ will not be part of the same colorful component in any feasible solution of MCC or MEC.

\fullversion{
Since we will consider MCC and MEC restricted to trees, we 
introduce some definitions that will be useful in the rest of the paper. 
Given a tree $G_T=(V,E)$ and a vertex $v \in V$, we denote by 
$G_T(v)$ the subtree of $G_T$ rooted at $v$.
The children of a node $v$ are called \emph{siblings}.
Moreover, we assume that for each internal vertex $v$ of a tree
the children of $v$ are ordered according to same ordering.
}

\paragraph{Parameterized Complexity}

A parameterized problem $(I,k)$ is said \textit{fixed-parameter tractable} (or in the 
class FPT) with respect to a parameter $k$ if it can be solved in $f(k)\cdot|I|^c$ time (in \textit{fpt-time}), where 
$f$ is any computable function and $c$ is a constant (see~\cite{Downey2013} 
for more details about fixed-parameter tractability).
The $O^*$ notation suppresses polynomial factors.
The class XP contains problems solvable in time $|I|^{f(k)}$, where $f$ is an unrestricted function.

A powerful technique to design parameterized algorithms is \textit{kernelization}. In short, 
kernelization is a polynomial-time self-reduction algorithm that takes an instance $(I,k)$ of a parameterized 
problem $P$ as input and computes an equivalent instance $(I',k')$ of $P$ such that 
$|I'| \leqslant h(k)$ for some computable function $h$ and $k' \leqslant k$. The instance $(I',k')$ 
is called a \textit{kernel} in this case. If the function $h$ is polynomial, we say that 
$(I',k')$ is a polynomial kernel. 

\appendixproof{}
{
A \textit{bikernelization} is a polynomial-time algorithm that maps an instance $(I,k)$ of a parameterized problem $P$ to an equivalent instance $(I',k')$ of a parameterized problem $P'$ (the bikernel) such that $|I'| \leq h(k)$ for some computable function $h$ and $k' \leq f(k)$. A kernelization is thus simply a bikernelization from $P$ to itself. Bikernelization was introduced 
in~\cite{DBLP:journals/algorithmica/AlonGKSY11}. 
}

Concerning approximation definitions, we refer the reader to some reference textbook like~\cite{Ausiello}.


\section{MCC for Trees: Parameterized Complexity and Approximability}
\label{sec:MCC:trees}

In this section, we show that MCC on trees is essentially equivalent to the 
\textsc{Minimum Multi-CUT} problem on Trees (M-CUT-T), thus the positive and negative results of (M-CUT-T) for parameterized complexity and approximability transfer to MCC.
We recall here the definition of M-CUT-T.

\Pb{Minimum Multi-CUT (M-CUT-T)}{a tree $T_M$ and a set $S_M$ of pairs of terminals.}{a minimum cut (that is a set of removed edges) such that, for each pair $(x,y) \in S_M$, $x$ and $y$ are disconnected through that cut.}{Output}

\subsection{Positive results}

We start by reducing  MCC to M-CUT-T, thus showing 
that MCC on trees admits an FPT algorithm (and a polynomial kernel) and 
a $2$-approximation algorithm .
We first describe the reduction. Given a colored tree $G_T=(V,E,C)$ as an instance of MCC, we define an instance $(T_M,S_M)$ of M-CUT-T 
as follows: $T_M$ is exactly $G_T$ (except for the colors of the vertices); 
for each pair $(x, y)$ of vertices in $G_T$ such that $c(x)=c(y)$, we define a
pair $(x,y)$ in $S_M$.
\fullversion
{We start by proving the following easy result.
\begin{lemma}
\label{lem:CutProp}
Consider a tree $G_T$ and suppose that $k$ edges of $G_T$ are cut. Then $G_T$ consists of $k+1$ connected components.
\end{lemma}
\begin{proof}
We prove the result by induction on $k$. If $k=0$ the lemma obviously holds.
Assume that, by inductive hypothesis, the lemma holds for at most $k$ edges cut, we prove that it holds for $k+1$ edges cut.
Consider one edge $\{u,v\}$ cut farthest from the root of $G_T$. Then the tree $G_T(v)$ contains no edge cut and one connected component. After the removal of $G_T(v)$ and $\{u,v\}$, the resulting tree $G'_{T}$ contains $k$ edges cut, and by 
inductive hypothesis, $k+1$ connected components. It follows that $G_T$, after $k+1$ edges are cut, contains $k+2$ 
connected components.
\end{proof}
}

Now, we prove the main lemma of this section.

\begin{lemma}
\label{lem:FPTMCC}
Consider an instance $G_T$ of MCC and the corresponding instance $(T_M,S_M)$ of M-CUT-T. Then: (1) given a solution of MCC on $G_T$ consisting of $k+1$ connected components, a solution of  M-CUT-T on $(T_M,S_M)$
consisting of $k$ edges cut can be computed in polynomial time; (2) given a solution of 
M-CUT-T on $(T_M,S_M)$ consisting of $k$ edges, a solution of MCC on $G_T$ consisting of $k+1$ connected components can be computed in polynomial time.
\end{lemma}
\begin{proof}
Consider a solution of MCC consisting of $k+1$ components obtained by removing a set $E'$ of $k$ edges. 
Then, $E'$ is a solution of M-CUT-T over instance $(T_M,S_M)$. Indeed, for each pair $(x,y) \in S_M$,
$c(x)=c(y)$, hence the two vertices belong to different connected components after the removal of edges in $E'$.

Conversely, consider a solution $E'$ of M-CUT-T over instance $(T_M,S_M)$, with $|E'|=k$. Then, remove the edges in $E'$ from $G_T$ and
consider the $k+1$ connected components induced by this removal in $G_T$. Since each pair $(x,y) \in S_M$ is disconnected
after the removal of $E'$, it follows that each connected component of $G_T$ after the removal of $E'$ is colorful.
\end{proof}

We can now easily give the main result of this section:
\begin{theorem}\label{th:mccpositive}
MCC when the input graph is a tree, MCC can be solved in time $O^*(1.554^k)$ where $k$ is the natural parameter and also admits a 2-approximation algorithm.
\end{theorem}
\appendixproof{Theorem~\ref{th:mccpositive}}
{
\begin{proof}
Since M-CUT-T can be solved in time $O^*(1.554^k)$~\cite{DBLP:journals/tcs/Kanj15}, by the property of our polynomial time
reduction and by Lemma~\ref{lem:FPTMCC}, it follows that MCC can be solved in time $O^*(1.554^k)$ on trees.

Moreover, M-CUT-T admits a factor $2$-approximation algorithm~\cite{DBLP:journals/algorithmica/GargVY97} on trees. 
Denote by $S(I)$ ($OPT(I)$, respectively) an approximation (optimal,
respectively) solution of and instance $I=(G_T)$ of M-CUT-T, 
and by $S(I')$ ($OPT(I')$, respectively) an approximation (optimal, respectively) solution of the corresponding instance $I'=(T_M,S_M)$ of MCC $(T_M,S_M)$. Then, by Lemma~\ref{lem:FPTMCC} and by the 2-approximation algorithm of M-CUT-T, it holds
\[
\frac{S(I')}{OPT(I')} = \frac{S(I)+1}{OPT(I)+1} \leq \frac{2OPT(I)+1}{OPT(I)+1}
\leq  
\]
\[\frac{2OPT(I)+2}{OPT(I)+1} = \frac{2OPT(I')}{OPT(I')}.
\]
Hence we can conclude that MCC admits a 2-approximation algorithm.
\end{proof}
}

Lemma~\ref{lem:FPTMCC} implies also a poly-kernel for MCC on trees.
\begin{theorem}
\label{thm:MCCbikernel}
If the input graph of MCC is a tree, it is possible to compute in polynomial time a kernel of size $O(k^3)$ where $k$ is the natural parameter.
\end{theorem}
\appendixproof{Theorem~\ref{thm:MCCbikernel}}
{
\begin{proof}
Consider the described reduction from an instance of $(G_T,k+1)$ of MCC 
to an instance $(T_M,k)$ of  M-CUT-T. 
Since there exists a kernel of size $O(k^3)$ for M-CUT-T~\cite{Chen2012}, 
it follows that, starting from an instance $(T_M,k)$ of M-CUT-T we can compute an instance $(T'_M,h)$ of M-CUT-T, with $h \leq k$, such that M-CUT-T on $T_M$ has a solution of size $k$ if and only if M-CUT-T on $T'_M$ has a solution of size $h$, and the number of vertices of $T'_M$ is bounded by $O(k^3)$.
Combining the two reductions, we have described a \emph{bikernel}, due to the following properties:
 
\begin{itemize}
\item An instance of $(G_T,k+1)$ of MCC admits a solution if and only if the corresponding instance $(T'_G,h)$
of M-CUT-T admits a solution of size $h$, where $(T'_G,h)$ can be computed in polynomial time
\item $h \leq k+1$
\item the number of vertices of $T'_G$ is bounded by $O(k^3)$
\end{itemize}
\end{proof}
}

\subsection{Lower bounds of MCC on trees}

Now we give a reduction from M-CUT-T to MCC on trees, thus proving a lower bound for the approximation of MCC on trees.
Starting from any instance $(T_M,S_M)$ of M-CUT-T, we compute a colored tree $G_T=(V,E,C)$, input of MCC, as follows. 
First, $G_T$ is isomorphic to $T_M$, and we color each vertex $v$ of $G_T$ with $c_v$. 
Denote by $E_1$ the edge set of such a tree.
Then, for each pair $(u,v) \in S_M$, we define a leaf $u_v$ adjacent to $v$ and colored $c_{u,v}$ and a leaf $v_u$ adjacent to $u$ and colored $c_{u,v}$ (see \autoref{fig:mcctrees}).
Denote by $E_2$ the edge set introduced by adding these edges.
More formally:

\begin{itemize}
\item $V = \{ V(T_M) \} \cup \{ u_v, v_u | (u,v) \in S_M \}$
\item $E_1 = E(T_M) $
\item $E_2 = \{ \{u_v,v\}, \{v_u,u\} | (u,v) \in S_M \} $
\item $E = E_1 \cup E_2$
\item $c(v) = c_v, \forall v \in V(T_M)$
\item $c(u_v) = c(v_u) = c_{u,v}, \forall (u,v) \in S_M$
\end{itemize}

\begin{figure}
\centering
\begin{tikzpicture}[scale=.9,transform shape,>=stealth,shorten <=.5pt,shorten >=.5pt]
\tikzstyle{vertex}=[circle,fill=black!0,minimum size=18pt,inner sep=0pt]

\node[vertex,draw,label=10:1] (m1) at (0,0) {};
\node[vertex,draw,label=10:2] (m2) at (-1,-1) {};
\node[vertex,draw,label={10:3}] (m3) at (0,-1) {};
\node[vertex,draw,label=10:4] (m4) at (1,-1) {};
\node[vertex,draw,label=10:5] (m5) at (-2,-2) {};
\node[vertex,draw,label=10:6] (m6) at (-1,-2) {};
\node[vertex,draw,label=10:7] (m7) at (0,-2) {};
\node[vertex,draw,label=10:8] (m8) at (1,-2) {};
\node[vertex,draw,label=10:9] (m9) at (1,-3) {};
\draw (m1) -- (m2);
\draw (m1) -- (m3);
\draw (m1) -- (m4);
\draw (m2) -- (m5);
\draw (m2) -- (m6);
\draw (m3) -- (m7);
\draw (m4) -- (m8);
\draw (m8) -- (m9);

\node () at (0,-4) {$T_M$};

\begin{scope}[xshift=6cm]
\node[vertex,draw] (m1) at (0,0) {$c_1$};
\node[vertex,draw] (m2) at (-1,-1) {$c_2$};
\node[vertex,draw] (m3) at (0,-1) {$c_3$};
\node[vertex,draw] (m4) at (1,-1) {$c_4$};
\node[vertex,draw] (m5) at (-2,-2) {$c_5$};
\node[vertex,draw] (m6) at (-1,-2) {$c_6$};
\node[vertex,draw] (m7) at (0,-2) {$c_7$};
\node[vertex,draw] (m8) at (1,-2) {$c_8$};
\node[vertex,draw] (m9) at (1,-3) {$c_9$};
\draw (m1) -- (m2);
\draw (m1) -- (m3);
\draw (m1) -- (m4);
\draw (m2) -- (m5);
\draw (m2) -- (m6);
\draw (m3) -- (m7);
\draw (m4) -- (m8);
\draw (m8) -- (m9);
\node[vertex,draw] (g57) at (-2.35,-3) {$c_{5,6}$};
\node[vertex,draw] (g58) at (-1.65,-3) {$c_{4,5}$};
\node[vertex,draw] (g6) at (-1,-3) {$c_{5,6}$};
\node[vertex,draw] (g2) at (-2,-1) {$c_{2,8}$};
\node[vertex,draw] (g4) at (2,-1) {$c_{4,5}$};
\node[vertex,draw] (g8) at (2,-2) {$c_{2,8}$};
\draw (m5) edge[ultra thick] (g58);
\draw (m5) edge[ultra thick] (g57);
\draw (m6) edge[ultra thick] (g6);
\draw (m2) edge[ultra thick] (g2);
\draw (m4) edge[ultra thick] (g4);
\draw (m8) edge[ultra thick] (g8);
\node () at (0,-4) {$G_T$};
\end{scope}
\end{tikzpicture}
\caption{Sample construction of $G_T$ from $T_M$ with $S_M = \{(2,8), (5,6), (4,5)\}$. Edge set $E_2$ of $T$ is drawn thick. For ease, colors of $G_T$ are drawn inside the nodes. On possible solution for this instance of M-CUT-T cuts edges $\{\{2,6\}, \{1,4\}\}$ and implies $3$ colorful connected components in the corresponding instance of MCC.}\label{fig:mcctrees}
\end{figure}

We start by proving a property of the tree $G_T=(V,E,C)$, input of MCC.

\begin{lemma}
\label{lem:MCCTressHard1}
Given a solution $G'=(V, E \setminus E')$ of MCC on $G_T=(V,E,C)$ consisting of $k$ colorful components, we can compute in polynomial time a solution $G''=(V, E \setminus E'')$ of MCC on $G_T=(V,E,C)$ consisting of at most $k$ colorful components such that $E'' \subseteq E_1$.
\end{lemma}
\begin{proof}
Consider the case that a (deleted) edge $\{ u,v \} \in E'$, where $v$ is a leaf introduced in $G_T$. 
Then, notice that the removal of edge $\{ u,v \}$ makes $v$ an isolated vertex. 
By construction $u$ and $v$ (and each leaf adjacent to $u$) have different colors. 
Hence there are two possible cases: either the colorful component $H$ that contains $u$ does not include vertices colored by $c_v$, hence we can add $v$ to $H$, thus we can avoid removing edge $\{ u,v \}$, or there is a vertex $w$ colored by $c_v$ in $H$.
In this case we can remove an edge of $E_1$, which separates $w$ from $u$ without removing edge $\{ u,v \}$; such an edge must exist, since $v$ and $w$ are leaves incident in different internal vertices. 
\end{proof}

Now, we prove that the reduction from M-CUT-T to MCC holds.

\begin{lemma}
\label{lem:MCCTressHard2}
Consider an instance $(T_M,S_M)$ of M-CUT-T and the corresponding instance $G_T=(V,E,C)$ of  MCC. Then: 
(1) given a solution of M-CUT-T over instance 
$(T_M,S_M)$ that cuts $k$ edges, we can compute in polynomial time a solution of 
MCC over instance $G_T=(V,E,C)$ consisting of at most $k+1$ colorful components;
(2) given a solution of MCC over instance $G_T=(V,E,C)$ consisting of at most $k+1$ colorful components,
we can compute in polynomial time a solution of M-CUT-T over instance 
$(T_M,S_M)$ that cuts at most $k$ edges.
\end{lemma}
\appendixproof{Lemma~\ref{lem:MCCTressHard2}}
{
\begin{proof}
(1). Consider a solution of M-CUT-T obtained by removing $k$ edges.
We compute a solution of MCC over instance $G_T=(V,E,C)$ by removing the corresponding edges of $E_1$. 
Now, since each pair $(u,v)$ in $S_M$ is disconnected in M-CUT-T, there exists a removed edge of $E_1$ on the unique path connecting two leaves adjacent to the vertices $u$ and $v$ and both colored by $c_{u,v}$. 
It follows that $u$ and $v$ belong to different connected components and that each connected component is colorful.
But then the solution of MCC consists of $k+1$ connected components.

(2). Consider a  solution of MCC over instance $G_T=(V,E,C)$ consisting of $k+1$ colorful
connected components. 
By Lemma~\ref{lem:MCCTressHard1}, it follows that the solution removes an edge set $E'_1 \subseteq E_1$. 
Now, consider the solution of M-CUT-T obtained by removing the edge set $E'_M$ corresponding to $E'_1$. 
It follows that each pair $(u,v)$ in $S_M$ is disconnected by removing $E'_M$, since the removal of edge set $E'_1$ from $G_T$ gives a graph consisting only of connected colorful components. 
Hence each pair of leaves having both color $c_{u,v}$ is disconnected.
Since the solution of  MCC over instance $G_T=(V,E,C)$ consists of $k+1$ colorful connected components and removes $k$ edges, it follows that the solution of M-CUT-T consists of $k$ removed edges.
\end{proof}
}

It was shown that M-CUT-T is as hard as \textsc{Minimum Vertex Cover} to approximate~\cite{DBLP:journals/algorithmica/GargVY97}, therefore, M-CUT-T cannot be approximated within factor 1.36 unless P=NP~\cite{Dinur2004} and within factor 2 assuming the Unique Game Conjecture (UGC)~\cite{Khot2008}.
Moreover, in the reduction given in~\cite{DBLP:journals/algorithmica/GargVY97}, the parameter is exactly the same, so  M-CUT-T cannot be solved in $2^{o(k)}n^{O(1)}$, assuming the ETH~\cite{LokshtanovMS11}.
Therefore, Lemma~\ref{lem:MCCTressHard1} and Lemma~\ref{lem:MCCTressHard2} allow to extend these results to MCC.

\begin{theorem}
\label{teo:APX-hardnessMCCTrees}
MCC on trees: (1) cannot be approximated within factor $1.36 -\varepsilon$, for any constant $\varepsilon>0$, unless P=NP, (2) cannot be approximated within factor $2 -\varepsilon$, for any constant $\varepsilon>0$, assuming the UGC and (3) cannot be solved in $2^{o(k)}n^{O(1)}$, assuming the ETH. 
\end{theorem}
\appendixproof{Theorem~\ref{teo:APX-hardnessMCCTrees}}
{
\begin{proof}
For (1), denote by $A(I)$ ($OPT(I)$, respectively) the value of an approximated (optimal, respectively) 
solution of M-CUT-T on instance $I=(T_M,S_M)$.
Denote by $A(I')$ ($OPT(I')$, respectively) the value of an approximated (optimal, respectively) 
solution of MCC on the corresponding instance $I'=(G_T)$. 
Then,
\[
\frac{A(I')}{OPT(I')} = \frac{A(I)+1}{OPT(I)+1} = 
\]
\[
= \frac{A(I)+1.36}{OPT(I)+1} -\frac{0.36}{OPT(I)+1}
\]

M-CUT-T cannot be approximated in factor 1.36 
(since it is hard to approximate as \textsc{Minimum Vertex Cover}~\cite{DBLP:journals/algorithmica/GargVY97}),
hence it holds, $A(I) \geq 1.36 OPT(I)$, which implies that

\[
A(I)+1.36 \geq 1.36(OPT(I)+1)
\]

It follows that
\[
\frac{A(I')}{OPT(I')} \geq 1.36 
-\frac{0.36}{OPT(I)+1}
\]
Defining $\varepsilon = \frac{0.36}{OPT(I)+1}$, the lemma holds, since if $\varepsilon$ is a constant,
then the same holds for $OPT(I)$ and $OPT(I')$, hence MCC can be trivially solved in constant time.

For (2), note that since M-CUT-T cannot is not approximable within
factor $2$ under the Unique Game Conjecture (UGC)~\cite{Khot2008}, 
thus the inequalities given above can modified by substituting $1.36$ with $2$,
showing that MCC on trees cannot be approximated within factor $2 -\varepsilon$, 
for any constant $\varepsilon>0$, assuming the UGC. 

For (3), observe that in our reduction the parameter increases only linearly and therefore preserves subexponential-time solvability.
\end{proof}
}

\section{Structural parameterization of MCC}
\label{sec:MCCDisjPath}


Since the MCC problem is already NP-hard on trees, we consider in this section the complexity of the problem when the input graph is a path or is close to a set of disjoint paths. 
We show that MCC can be easily solved in polynomial time when the input graph is a path (hence even when the input graph is a set of disjoint path), while, as a sharp contrast, MCC is not in the class XP for parameter distance to disjoint paths (more precisely, it is NP-hard even when the input graph is at distance 1 to the class of disjoint paths).

%


We start by showing that MCC on paths can be solved in polynomial time.

\begin{theorem}\label{thm:MCCpathspoly}
MCC on paths can be solved in $O(n^2)$-time.
\end{theorem}
\begin{proof} 
Assume that the input graph is a path $G_P=(V,E,C)$, and assume that the vertices on the path are ordered from $v_1$ to $v_n$. 
Define a function $M[j]$, with $0 \leq j \leq n$, as the minimum number of colorful components of a solution of MCC over instance $G_P$ restricted
to vertices $\{ v_1,\dots,v_j \}$.
$M[j]$, with $1 < j \leq n$, can be computed as follows:
\[ M[j] = \min_{0\leq t < j}
M[t] + 1, \text{ such that $v_{t+1}, \dots, v_j$ induce a colorful component.}
\]

In the base cases, that is when $j=0$ OR $j=1$, 
it holds $M[1]=1$, and $M[0]=0$. Next, we prove the correctness of the dynamic programming recurrence.

Given a path $G_P=(V,E,C)$ instance of MCC, 
there exists a solution of MCC on instance $G_P$ 
restricted to vertices $\{ v_1,\dots,v_j \}$ 
consisting of $h$ colorful components if and only if $M[j]=h$.
The base cases obviously holds, since $M[1]=1$ if and only if 
$v_1$ induces a colorful connected components and
$M[0]=0$ by definition.

We prove the lemma by induction on $j$.
Consider the case that $M[j]=h$, with $1 < j \leq n$
and $h \geq 1$.
Assume that $M[j]= M[t]+1$, for some $ 0 < t \leq j$. 
By induction hypothesis, assume that $t \geq 1$,
there exists a solution of MCC on instance $G_P$ restricted
to vertices $\{ v_1,\dots,v_t \}$  consisting of $h-1$ colorful components, thus there exists 
a solution of MCC on instance $G_P$ restricted
to vertices $\{ v_1,\dots,v_j \}$ consisting of $h$ colorful components.
If $t=0$, it holds $M[t] =0$, then $M[j]=h=1$.

Assume that there exists a solution of MCC on instance $G_P$ restricted
to vertices $\{ v_1,\dots,v_j \}$ consisting of $h$ connected components,
where $h \geq 0$.
Consider the colorful component that includes $v_j$, and assume 
that it is induced by
$v_{t+1}, \dots, v_j$, with $0 \leq t < j$. By induction hypothesis, 
it follows that $M[t]=h-1$, and that the connected component induced by $v_{t+1}, \dots, v_j$
is colorful, thus $M[j]=h$, concluding the proof.

It is then easy to see that the value of an optimal solution of MCC on path $G_P=(V,E,C)$ is stored in $M[n]$. 
The table $M[j]$ consists of $n$ entries and each entry can be computed in time $O(n)$, since we have to check at most $n$ value $t<j$, 
and the fact that the path $v_{t+1}, \dots, v_j$ is colorful can be precomputed in $O(n^2)$ time and then checked in constant time,
it follows that MCC on paths can be computed in time $O(n^2)$.
\end{proof}

Notice that if an input graph of MCC consists of disjoint paths, 
it can be solved in polynomial-time by applying the dynamic
programming algorithm independently to each path.

Now we prove that MCC is not in XP when parameterized by the
\emph{Distance to Disjoint Paths} number $d$ (the minimum number of 
vertices to remove from the input graph to have disjoint paths), even when 
the input graph is a tree.
We prove this result by giving a reduction 
from \textsc{Minimum Vertex Cover} (MinVC) to MCC on trees.

Consider an instance $G=(V,E)$ of MinVC, and let $G_C=(V_C,E_C)$ be the corresponding instance of MCC. $G_C$ is a rooted tree,  
defined as follows.
First, we define $|V|$ paths, one for each vertex in $G$. Path $P_i$ contains vertex $v_{c,i}$, colored by $c_i$, 
and vertices $e_{c,i,j}$, for each
$\{v_i,v_j\} \in E$, colored by $c_{ij}$. Notice that vertices $e_{c,i,j}$ appears in $P_i$ based on the lexicographic order
of the corresponding edges. 
Moreover, there exist two vertices associated with edge $\{v_i,v_j\} \in E$, namely $e_{c,i,j}$ (in $P_i$) and $e_{c,j,i}$ (in $P_j$), 
which are both colored by $c_{ij}$.
The tree $G_C$ is obtained by connecting the paths $P_1, \dots P_{|V|}$ to a root $r$, which is colored by $c_r$, where $c_r$ is a color
not associated with other vertices of $G_C$ (see \autoref{fig:vcmcc}).

\begin{figure}
\centering
\begin{tikzpicture}[scale=.6,transform shape]
\tikzstyle{vertex}=[circle,fill=black!0,minimum size=18pt,inner sep=0pt]
	\node[] () at (-11,7.5) {$G=(V,E)$};
    \node[vertex,draw,very thick,label=10:1] (IS1) at (-9,9) {};
    \node[vertex,draw,very thick,label=10:2] (IS2) at (-7,9) {};
	\node[vertex,draw,label=10:3] (IS3) at (-8,10) {};
    \node[vertex,draw,label=10:4] (IS4) at (-9,7.5) {};
    \node[vertex,draw,very thick,label=10:5] (IS5) at (-7,7.5) {};
     \path[draw] (IS1) -- (IS2) -- (IS3) -- (IS1);
     \path[draw] (IS2) -- (IS5) -- (IS4) -- (IS1);
	\begin{scope}[xshift=-9	cm,yshift=-4cm]
	     \node[] () at (5,11.0) {$G_C=(V_C,E_C)$};
	     \node[vertex,draw] (Tr) at (10,15) {$c_r$};
	     \node[vertex,draw] (c1) at (7,14) { $c_{1}$};
  	     \node[vertex,draw] (T11) at (7,13) { $c_{1,2}$};
	     \node[vertex,draw] (T12) at (7,12) { $c_{1,3}$};
	     \node[vertex,draw] (T13) at (7,11) { $c_{1,4}$};
	     \draw (Tr) edge[ultra thick] (c1);
	     \path[draw] (c1) -- (T11) -- (T12) -- (T13);
    	     \node[vertex,draw] (c2) at (8.5,14) { $c_{2}$};
  	     \node[vertex,draw] (T21) at (8.5,13) { $c_{1,2}$};
	     \node[vertex,draw] (T22) at (8.5,12) { $c_{2,3}$};
	     \node[vertex,draw] (T23) at (8.5,11) { $c_{2,5}$};
	     \path[draw] (Tr) edge[ultra thick] (c2) (c2) --  (T21) -- (T22) -- (T23);
	     \node[vertex,draw] (c3) at (10,14) { $c_{3}$};
  	     \node[vertex,draw] (T31) at (10,13) { $c_{1,3}$};
	     \node[vertex,draw] (T32) at (10,12) { $c_{2,3}$};
	     \path[draw] (Tr) -- (c3) -- (T31) -- (T32);
	     \node[vertex,draw] (c4) at (11.5,14) { $c_{4}$};
  	     \node[vertex,draw] (T41) at (11.5,13) { $c_{1,4}$};
	     \node[vertex,draw] (T42) at (11.5,12) { $c_{4,5}$};
	     \path[draw] (Tr) -- (c4) -- (T41) -- (T42);
		\node[vertex,draw] (c5) at (13,14) { $c_{5}$};
  	     \node[vertex,draw] (T51) at (13,13) { $c_{2,5}$};
	     \node[vertex,draw] (T52) at (13,12) { $c_{4,5}$};
	     \path[draw] (Tr) edge[ultra thick] (c5) (c5) -- (T51) -- (T52);
\end{scope}
\end{tikzpicture}
\caption{Sample construction of an instance of MCC from an instance of MinVC. A possible solution for MinVC is given in thick while edges to be cut for the instance of MCC are also in thick.}\label{fig:vcmcc}
\end{figure}
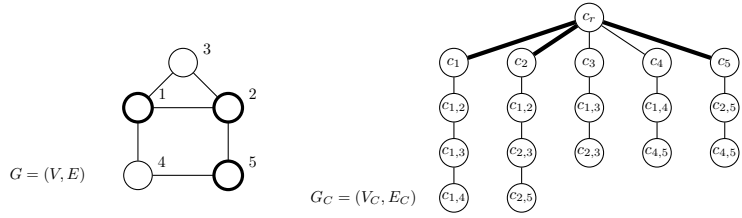

\begin{lemma}
\label{lem:MCCDisjPath1}
Let $G=(V,E)$ be an instance of MinVC, and let $G_C=(V_C,E_C)$ be the corresponding instance of MCC.
Then: (1) given a vertex cover of $G$ of size $k$, we can compute in polynomial time a solution of MCC over instance $G_C$
consisting of $k+1$ colorful components; (2) given a solution of MCC over instance $G_C$ consisting of $k+1$ colorful components,
we can compute in polynomial time a vertex cover of $G$ of size $k$.
\end{lemma}
\appendixproof{Lemma~\ref{lem:MCCDisjPath1}}
{
\begin{proof}
(1) Let $V' \subseteq V$ be a vertex cover of $G$, with $|V'|=k$. 
Define a solution $G'_C$ of MCC over instance $G_C$ as follows.
For each vertex $v_i \in V'$, remove the edge $\{ r,v_{c,i} \} \in E_C$
such that $P_i$ becomes a connected component
disconnected from $r$. 
Notice that the graph $G'_C$ consists of $k+1$ connected components. 
Moreover, each connected component of $G'_C$ is colorful. 
Indeed, each $P_i$ is colorful by construction. 
Consider the component $T$ containing the root $r$.
Notice that $T$ is colorful, since if two paths $P_i$ and $P_j$ are connected to $r$, then, by the property
of $V'$, $\{v_i,v_j\} \notin E$.

(2) Let $G'_C$ be a solution of MCC over instance $G_C$ 
consisting of $k+1$ colorful components. 
Denote by $P'_i$ the path consisting of $r$ and path $P_i$.
We construct a solution $G^*_C$ of MCC over instance $G_C$ 
consisting of at most $k+1$ colorful components as follows: 
if an edge of $P'_i$ is removed to obtain $G'_C$, $G^*_C$ is obtained
by cutting edge $\{ r,v_{c,i} \}$. 
$G^*_C$ consists of at most $k+1$ connected components, 
since it is obtained by removing no more edges than $G'_C$.  
Notice that each connected component of $G^*_C$ is colorful. 
Indeed, again each $P_i$ is colorful by construction. 
Furthermore, consider the colorful component $T$ containing the root $r$, and 
the paths $P'_i$ and $P'_j$,  with $\{v_i,v_j\} \in E$. By construction both paths contain a vertex 
colored by $c_{ij}$, hence one edge of the paths $P'_i$ or $P'_j$ must be removed. Hence 
$G'_C$ is obtained by cutting at least one edge in 
$P'_i$ or $P'_j$, thus, by construction, $T$ is colorful.
\end{proof}
}

By the previous lemma, the following result holds.

\begin{theorem}
\label{thm:MCCDIsPath}
MCC is NP-hard even when the input graph is at distance $1$ to Disjoint Paths. 
\end{theorem}
\appendixproof{Theorem~\ref{thm:MCCDIsPath}}
{
\begin{proof}
Notice that the graph $G_C$ has distance $1$ to Disjoint Path, since it is enough to remove the root to obtain
$|V|$ disjoint paths. Moreover, by Lemma~\ref{lem:MCCDisjPath1} and, by the NP-hardness of MinVC, 
the result follows.
\end{proof}
}


It is worth noticing that this result extends to parameter pathwidth or distance to interval graph, as these last parameters are ``stronger'' than distance to disjoint path in the sense of~\cite{DBLP:conf/mfcs/KomusiewiczN12}. 

\section{An FPT Algorithm for MEC Parameterized by $k$}
\label{sec:FPTMEC}

We present a parameterized algorithm for MEC with respect to the natural parameter $k$, that is the number of colorful connected component. 
\fullversion{Whereas one could obtain a parameterized algorithm using the color-coding technique~\cite{DBLP:journals/jacm/AlonYZ95} without much difficulty, 
we will show that the problem admits an exponential size kernel, which implies that the problem is in FPT.}

Given a colored graph $G$, we first compute a Depth-First-Search (DFS) $D=(V,E_D,E_B)$ of $G$.
Recall that a DFS $D$ of $G$ consists of a tree induced by
$D'=(V,E_D)$ (hence not considering edges in $E_B$), while
$E_B= E \setminus E_D$ are called \emph{backward edges} and have the following well-known property
(see~\cite{DBLP:books/daglib/0023376} for details).

\begin{lemma}
\label{lem:backedges}
Consider a graph $G$ and a DFS $D=(V,E_D,E_B)$ of $G$.
Let $\{u,v\} \in E_B$ be a backward edge. Then 
$u$ and $v$ are on a path from a leaf of $D'$ to the root
of $D'$.
\end{lemma}

We will first consider some easy cases 
where there is a solution of MEC of size at least $k$. 
Let $V_A$ be the set of vertices of $V$ which are parent of a leaf in $D'$. 
The following properties holds.

\begin{lemma}
\label{lem:pathbounded}
If there exists a path in $D'$ from the root $r(D')$ to a leaf of $D'$ of length at least $2k$, then
there exists a solution of MEC of size at least $k$. 
\end{lemma}
\appendixproof{Lemma~\ref{lem:pathbounded}}
{
\begin{proof}
Consider a path of length at least $2k$ from $r(D')$ to a leaf of $D'$. 
It follows that there exists a matching in $D'$ (hence also in $G$)
consisting of at least $k$ edges, and the lemma follows.
\end{proof}
}

\begin{lemma}
\label{lem:boundVA}
If $|V_A| \geq k$, then 
there exists a solution of MEC of size at least $k$.
\end{lemma}
\appendixproof{Lemma~\ref{lem:boundVA}}
{
\begin{proof}
Consider a vertex $v \in V_A$ and a leaf $l$ of $D'$ adjacent to $v$ in $D'$. 
Then define a colorful component induced by $v$ and $l$. It follows that there exist 
at least $k$ colorful component in $D'$, hence in $G$, and the lemma follows.
\end{proof}
}

Now, for each vertex $v \in V_A$ we consider the leaves adjacent to $v$ and their colors. 
Define the set $C_x(v)$ as  the set of leaves colored by $c_x$ and adjacent to $v \in V_A$ in $D'$. Formally,
\fullversion{
\[
C_x(v) = \{ l: \text{there exists a leaf $l$ colored by $c_x$ adjacent to $v$}\} 
\]
}
Then the following property holds.

\begin{lemma}
\label{lem:boundColor}
Given a vertex $v \in V_A$, if there exist $\sqrt{2k}$ non-empty sets $C_x(v)$ associated with
distinct colors $c_x$, then there exists a solution of MEC of size at least $k$. 
\end{lemma}
\appendixproof{Lemma~\ref{lem:boundColor}}
{
\begin{proof}
Assume that there exist $\sqrt{2k}$ non empty sets $C_x(v)$ for different colors. 
Then, define a colorful component consisting of $v$ and one vertex for each set 
$C_x(v)$.  It follows that the component consists of at least $\sqrt{2k}+1$ vertices,
hence its transitive closure contains at least $k$ edges.
\end{proof}
} 
 
%
Consider vertex $u \in C_x(v)$, for some $v \in V_A$, and 
define the following set of vertices:
\[
Adj(u)=\{w \in V: \{u,w \} \in E \} 
\]

Moreover, define the following collection $Adj(C_x(v))$ of sets of vertices:
\[
Adj(C_x(v)) = \{Adj(u): u \in C_x(v) \}  
\]

The following property holds.

\begin{lemma}
\label{lem:boundbackconnection}
Given a vertex-colored graph $G$ such that the hypothesis of Lemma~\ref{lem:pathbounded} 
does not hold,
consider a vertex $v$ in $V_A$ and a set $C_x(v)$. 
Then $|Adj(C_x(v))| \leq 2^{2k+1}$.
\end{lemma}
\appendixproof{Lemma~\ref{lem:boundbackconnection}}
{
\begin{proof}
Consider the vertices in $C_x(v)$. 
By construction each of such vertex is adjacent to exactly one vertex in $D'$; moreover, we claim that
each vertex $l$ in $C_x(v)$ is adjacent to at most $2k+1$ vertices in $D$. 
Indeed, if $l$ is adjacent to more than $2k+1$ vertices in $D$, then there exist $2k$ vertices
on the path from the root of $D$ to $l$ such that $l$ is connected to these vertices via backward edges.
Then there exists a path in $D'$ from the root $r(D')$ to a leaf of $D'$ of length at least $2k$ 
and Lemma~\ref{lem:pathbounded} holds. 
Hence, it holds that each vertex $l$ in $C_x(v)$ is adjacent to at most $2k+1$ vertices in $D$.
But then, the number of possible subsets
of vertices adjacent to a vertex in $C_x(v)$ is bounded by $2^{2k+1}$,
hence $|Adj(C_x(v))| \leq 2^{2k+1}$.
\end{proof}
}

Based on Lemma~\ref{lem:boundbackconnection}, we can partition the vertices of each $C_x(v)$ into sets $C_{x,1}(v), \dots, C_{x,p}(v)$, with $p \leq 2^{2k+1}$, based on the fact that two vertices of $C_x(v)$ belong to the same set $C_{x,t}(v)$ if they have the same set of adjacent vertices.
Since by Lemma~\ref{lem:boundbackconnection} $|Adj(C_x(v))| \leq 2^{2k+1}$, the number of possible subsets of $C_x(v)$ is at most $2^{2k+1}$, hence $p \leq 2^{2k+1}$.


Now, assume that the hypotheses of Lemma~\ref{lem:pathbounded}, Lemma~\ref{lem:boundVA} and Lemma~\ref{lem:boundColor} do not hold.
Consider an algorithm that, for each set $C_{x,i}(v)$, 
computes a set $C'_{x,i}(v)$ by picking at most $k$ vertices of $C_{x,i}(v)$ and removing the other vertices of $C_{x,i}(v)$. 
Let $G'$ be the resulting graph. 
We claim that $G'$ contains at most $O(k^2 2^{2k+1})$ vertices.
First, notice that each $C'_{x,t}(v)$ contains at most $k$ vertices and that, for each vertex $v$,
there exists at most $2^{2k+1}$ sets $C'_{x,t}(v)$. 
Since, there exist at most $O(k \sqrt{k})$ sets $C_x(v)$
(at most $ \sqrt{2k}$ colors $c_x$ and at most 
$k$ vertices $v \in V_A$), we can conclude that $G'$ contains
at most $O(k^2\sqrt{k} 2^{2k+1})$ vertices in sets $C'_{x,i}(v)$. 

Now, consider the vertices $G'$ which are not contained in some set $C'_{x,i}(v)$.
These vertices correspond to internal vertices of $D'$. Since the hypothesis of Lemma~\ref{lem:pathbounded} does not hold,
$D'$ is a tree of depth at most $2k$, 
and there exist at most $k$ vertices adjacent to leaves, as $|V_A| <k$. Hence there exist at most $k$ paths of length $2k$ in $D'$ 
from the root to vertices adjacent to leaves, thus we can conclude that there exist at most $2k^2$ internal vertices in $D'$.
Hence there exists at most $2k^2$ vertices in $G'$ which are not contained in some set  $C'_{x,i}(v)$.

Now, we prove that $(G',k)$ is a kernel for MEC.


\begin{lemma}
\label{lem:expkernel}
There exists a collection of disjoint colorful components $V_1, \dots , V_h$
of size at least $2$ in $G$ 
if and only if there exists a collection of disjoint colorful components $V'_1, \dots , V'_h$ in $G'$,
with $|V_i|=|V'_i|$, $1 \leq i \leq h$.
\end{lemma}
\appendixproof{Lemma~\ref{lem:expkernel}}
{
\begin{proof}
Obviously if there exists a collection of disjoint colorful components $V'_1, \dots , V'_h$ of size at least $2$ in $G'$,
then there exists a collection of disjoint 
colorful components $V_1, \dots , V_h$ in $G$ with 
$|V_i|=|V'_i|$, $1 \leq i \leq h$.

Now, consider a collection of disjoint colorful components $V_1, \dots , V_h$ of size at least $2$ in $G$.
Notice that for each $V_i$ at most one vertex can be in some sets $C_{x,i}(v)$ and that,
if $|C'_{x,i}(v)|=t \leq k$, at most $t$ colorful components in $V_1, \dots , V_h$ can
contain a node in $C_{x,i}(v)$.

Now, we compute $V'_1, \dots , V'_h$ as follows. For each 
$C'_{x,i}(v)$ partition its vertices assigning a vertex to $V_j'$ if and only if $V_j$ contains a
vertex in $C_{x,i}(v)$. 
Then partition the internal vertices of $D$ as they are partitioned by $V_1, \dots V_h$, that is assign 
vertex $u$ to $V'_j$ if and only if $u \in V_j$.

Now, by construction $V'_1, \dots , V'_h$ are disjoint and $|V_i|=|V'_i|$, for each $1 \leq i \leq h$.
Moreover, each $V'_i$ is colorful, since $V_i$ is colorful and we have added to $V'_i$ vertices
having the same colors as those of $V_i$.
Finally, notice that each $V'_i$ is a connected component. 
First, notice that the leaves of $D$ are not adjacent by the property of DFS, and that 
they are only connected to internal vertices of $D$.
Now, consider $V_i$ and $V'_i$. The 
two components contain the same subset of internal vertices of $D$; for each vertex $v_i$ in $V_i$ there is a corresponding
vertex $v'_i$ in $V'_i$ that is connected to the same set of vertices of $D$. Then, since $V_i$ is a connected component,
also $V'_i$ is connected component.
\end{proof}
}

Hence we have the following result.

\begin{theorem}
\label{th:expo kernel mec}
There exists a kernel of size $O(k^2\sqrt{k} 2^{2k+1})$ for MEC.
\end{theorem}
\appendixproof{Theorem~\ref{th:expo kernel mec}}
{
\begin{proof}
The result follows from Lemma~\ref{lem:expkernel} and from the fact that graph $G'$ contains 
at most $k^2$ internal vertices of $D$ (by Lemma~\ref{lem:pathbounded} and by Lemma~\ref{lem:boundVA})
and at most $O(k \sqrt{k}2^{2k+1})$ sets $C_{x,i}(v)$ (by Lemma~\ref{lem:boundbackconnection} 
and by Lemma~\ref{lem:boundColor}), each of size bounded by $k$.
\end{proof}
}

\section{A poly-kernel for MEC on trees}
\label{sec:kernelMECtrees}

In this section, we show that in the special case of MEC where the input graph is a tree, the kernel size can be quadratic.
The algorithm is similar to the one of Section~\ref{sec:FPTMEC}.
Consider a colored tree $G_T=(V,E,C)$, and let $r(G_T)$ denote the root of $G_T$. 
Lemmata~\ref{lem:pathbounded},\ref{lem:boundVA},\ref{lem:boundColor} hold for $G_T$.
Hence, we focus only on the leaves of $G_T$. 


Since $G_T$ is a tree, it follows that a leaf $u$ having ancestor $v$ belongs 
to a component of size at least $2$ only if $u$ and $v$ belongs to the same component. 
It follows that among the leaves having the same color $c_x$ and adjacent to
a vertex $u$, only one can belong to a colorful component of size at least $2$. 
Hence, given $v \in V_A$,
let $C_x(v)$ be the set of leaves adjacent to $v$ and colored by $c_x$. We remove all but one vertex from $C_x(v)$.
Let $G_T'$ be the resulting tree. We have the following property for $G_T'$.

\begin{lemma}
\label{lem:polykernel}
\sloppy
There exists a collection of disjoint colorful components $V_1, \dots , V_h$ of size at least $2$ in $G_T$ 
if and only if there exists a collection of disjoint colorful components $V'_1, \dots , V'_h$ in $G_T'$,
with $|V_i|=|V'_i|$, $1 \leq i \leq h$.
\end{lemma}
\appendixproof{Lemma~\ref{lem:polykernel}}
{
\begin{proof}
Obviously if there exists a collection of disjoint colorful components $V'_1, \dots , V'_h$ in $G_T'$ of size at least $2$,
then then there exists a collection of disjoint colorful components $V_1, \dots , V_h$ in $G_T$ with 
$|V_i|=|V'_i|$, $1 \leq i \leq h$.

For the reverse direction, consider a colorful component $V_i$ of $G_T$ of size at least $2$ containing vertex $v \in V_A$. 
We compute a corresponding component $V'_i$ of $G_T'$ having the same size of $V_i$ as follows.
We add all the internal vertices of $V_i$ to $V'_i$; for each color $c_x$ such that there exists a leaf $u$ in $V_i$ adjacent to $v$ and colored by $c_x$, we add the only vertex of set $C_x(v)$ to $V'_i$.

By construction, since the components $V_1, \dots , V_h$ 
are colorful, the same property holds for 
components $V'_1, \dots V'_h$.
Moreover, since components $V_1, \dots V_h$ are disjoint, the same property holds for components $V'_1, \dots V'_h$, 
as each vertex $v \in V_A$ belongs only to one connected component $V_i$, and since $G_T$ is a tree, 
the same property holds for each leaf adjacent to $v$.
\end{proof}
}

\begin{theorem}
\label{th:poly kernel trees}
There exists a kernel of size $O(k^2)$ for MEC on trees.
\end{theorem}
{
\begin{proof}
The correctness of the construction of $G_T'$ follows from Lemma~\ref{lem:polykernel}.
$G_T'$ contains at most $k^2$ internal vertices (by Lemma~\ref{lem:pathbounded} and by Lemma~\ref{lem:boundVA}).
Moreover, by Lemma~\ref{lem:boundColor}, $G_T'$ contains at most $O(k \sqrt{k})$ sets $C_{x,i}(v)$ (recall that these sets are defined such that two vertices of $C_x(v)$ are in the same set $C_{x,i}(v)$ if they have the same set of adjacent vertices), each of size bounded by $1$.
Hence the total number of vertices is bounded by $O(k^2)$ and the total number of edges, since $G_T'$ is a tree is bounded by $O(k^2)$.
\end{proof}
}

\section{Structural parameterization of MEC}
\label{sec:MECDisjPath}

In this section, we consider the MEC problem restricted
to paths and graph at bounded distance from disjoint path.
We show that, after appropriate modifications, 
the results on structural parameterization for MCC hold also for MEC .

\begin{theorem}
\label{thm:mecpathpoly}
MEC on paths can be solved in $O(n^2)$-time.
\end{theorem}
\appendixproof{Theorem~\ref{thm:mecpathpoly}}
{
\begin{proof}
Define $M[j]$ as the minimum number of colorful components of a solution of MEC over instance $G_P$ restricted
to vertices $\{ v_1,\dots,v_j \}$.
$M[j]$, with $j>1$, can be computed as follows:
\begin{small}
\[ M[j] = \min_{0\leq t < j}
M[t] + \frac{(j-t-1)(j-t)}{2}, \text{ s.t. $v_{t+1}, \dots, v_j$ induce a colorful component}
\]
\end{small}

In the base cases, it holds $M[1]=0$, and $M[0]=0$. 
Next, we prove the correctness of the dynamic programming recurrence.

We claim that given a path $G_P=(V,E,C)$ instance of MEC, there exists a solution of MEC on instance $G_P$ restricted to vertices $\{ v_1,\dots,v_j \}$ with a transitive closure consisting of $h$ edges if and only if $M[j]=h$.
The base cases obviously hold.

We prove the claim by induction on $j$.
Consider the case that $M[j]=h$ and assume that 
$M[j]= M[t]+\frac{(j-t+1)(j-t)}{2}$, for some $0 \leq t \leq j$. 
By induction hypothesis, if $t>0$
there exists a solution of MEC on instance $G_P$ restricted to vertices 
$\{ v_1,\dots,v_t \}$ having a transitive closure consisting of $h- \frac{(j-t+1)(j-t)}{2}$ edges, thus there exists  a solution of MEC on instance $G_P$ restricted to vertices $\{ v_1,\dots,v_j \}$ having a transitive closure consisting of $h$ edges.
If $t=0$, $M[0]=0$, thus $M[j] = \frac{(j-t+1)(j-t)}{2} = h$.

Assume that there exists a solution of MEC on instance $G_P$ restricted to vertices $\{ v_1,\dots,v_j \}$ having a transitive closure consisting of $h$ edges.
Consider the colorful component that includes $v_j$, and assume that it is induced by vertices $v_{t+1}, \dots, v_j$, with $0 \leq t \leq j$. 
By induction hypothesis, it follows that $M[t]=h-\frac{(j-t+1)(j-t)}{2}$, and furthermore that the connected component induced by $v_{t+1}, \dots, v_j$ is colorful, thus $M[j]=h$, proving the claim.

It is easy to see that the value of an optimal solution of MEC on path $G_P=(V,E,C)$ is stored in $M[n]$. 
Since the table $M[j]$ consists of $n$ entries and each entry can be computed in time $O(n)$ as for MCC, it follows that MEC on paths can be computed in time $O(n^2)$.
\end{proof}
}

Similarly to MCC, MEC is NP-hard even if we restrict the instance to graphs having distance $1$ to Disjoint Paths. As for MCC, it is worth noticing that this hardness result extends to other stronger parameters like pathwidth~\cite{DBLP:conf/mfcs/KomusiewiczN12}.

\begin{theorem}
\label{thm:MECDIsPath}
MEC is NP-hard even when the input graph has distance $1$ to Disjoint Paths. 
\end{theorem}
\appendixproof{Theorem~\ref{thm:MECDIsPath}}
{
\begin{proof}
The results follows from a proof similar to that of Section~\ref{sec:MCCDisjPath}.
We prove that MEC is not in XP when parameterized by the 
\emph{Distance to Disjoint Paths} number $d$, 
even when the input graph is a tree.
We give a reduction from \textsc{Maximum Independent Set} (MaxIS) to MEC on trees.

Consider an instance $G=(V,E)$ of MaxIS, and let $G_C=(V_C,E_C)$ be the corresponding instance of MEC. $G_C$ is a rooted tree,  
defined as follows.
First, we define $|V|$ paths, one for each vertex in $G$. Path $P_i$ contains vertex $v_{c,i}$, colored by $c_i$, 
and vertices $e_{c,i,j}$, for each
$\{v_i,v_j\} \in E$, colored by $c_{ij}$, followed by a path $P_{A,i}$ so that $P_i$ consists of 
$n^3$ vertices, each colored with a distinct color $c_{a,i}$. 
Notice that vertices $e_{c,i,j}$ appears in $P_i$ based on the lexicographic order
of the corresponding edges of $G$. 
Moreover, there exist two vertices associated with edge $\{v_i,v_j\} \in E$, namely $e_{c,i,j}$ (in $P_i$) and $e_{c,j,i}$ (in $P_j$), 
which are both colored by $c_{ij}$.
The tree $G_C$ is obtained by connecting the paths $P_1, \dots P_{|V|}$ to a root $r$, which is colored by $c_r$, 
where $c_r$ is a fresh new color.

\begin{lemma}
\label{lem:MECDisjPath1}
Let $G=(V,E)$ be an instance of MaxIS, and let $G_C=(V_C,E_C)$ be the corresponding instance of MEC.
Then (1) given an independent set of $G$ of size $k$, we can compute in polynomial time a solution of MEC over instance $G_C$
having transitive closure of size at least $\frac{kn^3(kn^3+1)}{2}+(n-k)\frac{n^3(n^3-1)}{2}$; 
(2) given a solution of MEC over instance $G_C$ having  a transitive closure of size $\frac{kn^3(kn^3+1)}{2}+(n-k)\frac{n^3(n^3-1)}{2}$,
we can compute in polynomial time an independent set of $G$ of size at least $k$.
\end{lemma}
\begin{proof}
(1) Consider an independent set $V'$ of $G$, with $|V'|=k$. Define a solution $G'_C$ of MEC over instance $G_C$ as follows.
For each vertex $v_i \in V\setminus V'$, cut the edge $\{ r,v_{c,i} \} \in E_C$ such that $P_i$ becomes a connected component
disconnected from $r$. Notice that each connected component is colorful. 
Indeed, each $P_i$ is colorful by construction. Consider the component $T$ containing the root $r$.
Notice that $T$ is colorful, since if two paths $P_i$ and $P_j$ are connected to $r$, then, by the property
of $V'$, $\{v_i,v_j\} \notin E$.
Moreover, the connected component that includes $r$ contains $kn^3+1$ vertices, hence it has transitive closure 
of size at least $\frac{kn^3(kn^3-1)}{2}$. Each of the other $(n-k)$ component consists of 
$\frac{n^3(n^3-1)}{2}$ edges.

(2) Consider a solution $G'_C$ of MEC over instance $G_C$ having a transitive closure of size at least 
$\frac{kn^3(kn^3-1)}{2}+(n-k)\frac{n^3(n^3-1)}{2}$.
We claim that the connected component $K_R$ including the root contains at least $k$ paths $P_i$. 
Notice that we assume that
if the first vertex of $P_{A,i}$ (the one connected to a vertex not in $P_{A,i}$) belongs to $K_R$, 
then we can easily extend the solution so that $K_R$ includes every vertex of $P_i$.
Assume to the contrary that $K_R$ includes $h<k$ path $P_i$. 
It follows that
$K_R$ contains at most $hn^3+n^2+1$ vertices. It follows that the transitive closure of such a solution
contains at most $\frac{(hn^3+n^2)(hn^3+n^2+1)}{2}+(n-h)\frac{n^3(n^3-1)}{2}$ edges. 

Since $h^2n^6+2hn^5+hn^3+1+h(n^3(n^3-1))<k^2n^6+kn^3+k(n^3(n^3-1)$,
for a value of $n$ sufficiently large, it holds that 
$\frac{(hn^3+n^2)(hn^3+n^2+1)}{2}+(n-h)\frac{n^3(n^3-1)}{2}<\frac{kn^3(kn^3-1)}{2}+(n-k)\frac{n^3(n^3-1)}{2}$.

Hence $K_R$ must include $k$ path $P_i$. 
Denote by $P'_i$ the path consisting of $r$ and path $P_i$.
Consider the
the paths $P'_i$ and $P'_j$,  with $\{v_i,v_j\} \in E$. By construction both paths contain a vertex 
colored by $c_{ij}$, hence one edge of the paths $P'_i$ and $P'_j$ must be cut
and $P_i$, $P_j$ cannot both be part of $K_R$.
Hence, we can define an independent set $V'$ of $G$ as follows:

\[
V'=\{  v_i:\text{ $P_i$ belongs to $K_R$} \}
\]
\end{proof}

Notice that the graph $G_C$ has distance $1$ to Disjoint Path, since it is enough to remove the root to obtain
$|V|$ disjoint paths. Moreover, by Lemma~\ref{lem:MECDisjPath1} and, by the NP-hardness of MaxIS, the result follows.

\end{proof}
}

\section{Conclusion}
We have considered two variants of the problem of finding colorful components inside a graph, and we have studied their parameterized and approximation complexity, for general and restricted instances. 
In the future, we aim at refining the parameterized complexity analysis, for example deepen the 
structural results for MCC and MEC.
Moreover, it would be interesting to study the parameterized complexity of the two problems under other meaningful parameters
in the direction of parameterizing above a guaranteed value~\cite{DBLP:journals/jal/MahajanR99}.
For example, in the case of MEC, one could compute in polynomial time a matching $M$ of the input graph. 
Since no edge with both endpoints with the same color exists, this matching is a feasible solution of MEC. 
As a consequence, the optimum solution is always bigger than $|M|$. 
The parameterized complexity of this problem with respect to the parameter "difference between the optimum and the size of the matching" could be interesting and more meaningful than the parameter value of the optimum since it informally represents the "hard" part of the problem.



\section*{References}

\bibliographystyle{elsarticle-num}

\bibliography{biblio}

\begin{thebibliography}{10}
\expandafter\ifx\csname url\endcsname\relax
  \def\url#1{\texttt{#1}}\fi
\expandafter\ifx\csname urlprefix\endcsname\relax\def\urlprefix{URL }\fi
\expandafter\ifx\csname href\endcsname\relax
  \def\href#1#2{#2} \def\path#1{#1}\fi

\bibitem{DBLP:journals/tcbb/LacroixFS06}
V.~Lacroix, C.~G. Fernandes, M.~Sagot, Motif search in graphs: Application to
  metabolic networks, {IEEE/ACM} Trans. Comput. Biology Bioinform. 3~(4) (2006)
  360--368.
\newblock \href {http://dx.doi.org/10.1109/TCBB.2006.55}
  {\path{doi:10.1109/TCBB.2006.55}}.

\bibitem{DBLP:conf/cpm/BrucknerHKNTU12}
S.~Bruckner, F.~H{\"{u}}ffner, C.~Komusiewicz, R.~Niedermeier, S.~Thiel,
  J.~Uhlmann, Partitioning into colorful components by minimum edge deletions,
  in: J.~K{\"{a}}rkk{\"{a}}inen, J.~Stoye (Eds.), Combinatorial Pattern
  Matching - 23rd Annual Symposium, {CPM} 2012, Vol. 7354 of LNCS, Springer,
  2012, pp. 56--69.
\newblock \href {http://dx.doi.org/10.1007/978-3-642-31265-6_5}
  {\path{doi:10.1007/978-3-642-31265-6_5}}.

\bibitem{DBLP:journals/tcs/DondiFV13}
R.~Dondi, G.~Fertin, S.~Vialette, Finding approximate and constrained motifs in
  graphs, Theor. Comput. Sci. 483 (2013) 10--21.
\newblock \href {http://dx.doi.org/10.1016/j.tcs.2012.08.023}
  {\path{doi:10.1016/j.tcs.2012.08.023}}.

\bibitem{DBLP:journals/tcbb/BetzlerBFKN11}
N.~Betzler, R.~van Bevern, M.~R. Fellows, C.~Komusiewicz, R.~Niedermeier,
  Parameterized algorithmics for finding connected motifs in biological
  networks, {IEEE/ACM} Trans. Comput. Biology Bioinform. 8~(5) (2011)
  1296--1308.
\newblock \href {http://dx.doi.org/10.1109/TCBB.2011.19}
  {\path{doi:10.1109/TCBB.2011.19}}.

\bibitem{DBLP:journals/jda/DondiFV11}
R.~Dondi, G.~Fertin, S.~Vialette, Complexity issues in vertex-colored graph
  pattern matching, J. Discrete Algorithms 9~(1) (2011) 82--99.
\newblock \href {http://dx.doi.org/10.1016/j.jda.2010.09.002}
  {\path{doi:10.1016/j.jda.2010.09.002}}.

\bibitem{DBLP:conf/wabi/ZhengSLS11}
C.~Zheng, K.~M. Swenson, E.~Lyons, D.~Sankoff, {OMG! Orthologs in Multiple
  Genomes - Competing Graph-Theoretical Formulations}, in: T.~M. Przytycka,
  M.~Sagot (Eds.), Algorithms in Bioinformatics - 11th International Workshop,
  {WABI} 2011, Vol. 6833 of LNCS, Springer, 2011, pp. 364--375.
\newblock \href {http://dx.doi.org/10.1007/978-3-642-23038-7_30}
  {\path{doi:10.1007/978-3-642-23038-7_30}}.

\bibitem{DBLP:journals/algorithmica/AdamaszekP15}
A.~Adamaszek, A.~Popa, Algorithmic and hardness results for the colorful
  components problems, Algorithmica 73~(2) (2015) 371--388.
\newblock \href {http://dx.doi.org/10.1007/s00453-014-9926-0}
  {\path{doi:10.1007/s00453-014-9926-0}}.

\bibitem{DBLP:conf/iwoca/AdamaszekBP14}
A.~Adamaszek, G.~Blin, A.~Popa, Approximation and hardness results for the
  maximum edges in transitive closure problem, in: J.~Kratochv{\'{\i}}l,
  M.~Miller, D.~Froncek (Eds.), Combinatorial Algorithms - 25th International
  Workshop, {IWOCA} 2014, Vol. 8986 of LNCS, Springer, 2014, pp. 13--23.
\newblock \href {http://dx.doi.org/10.1007/978-3-319-19315-1_2}
  {\path{doi:10.1007/978-3-319-19315-1_2}}.

\bibitem{DBLP:conf/cie/DondiS16}
R.~Dondi, F.~Sikora,
  \href{https://doi.org/10.1007/978-3-319-40189-8_27}{Parameterized complexity
  and approximation issues for the colorful components problems}, in:
  A.~Beckmann, L.~Bienvenu, N.~Jonoska (Eds.), Pursuit of the Universal - 12th
  Conference on Computability in Europe, CiE 2016, Paris, France, June 27 -
  July 1, 2016, Proceedings, Vol. 9709 of Lecture Notes in Computer Science,
  Springer, 2016, pp. 261--270.
\newblock \href {http://dx.doi.org/10.1007/978-3-319-40189-8_27}
  {\path{doi:10.1007/978-3-319-40189-8_27}}.
\newline\urlprefix\url{https://doi.org/10.1007/978-3-319-40189-8_27}

\bibitem{Downey2013}
R.~G. Downey, M.~R. Fellows, Fundamentals of Parameterized Complexity,
  Springer-Verlag, 2013.

\bibitem{DBLP:journals/algorithmica/AlonGKSY11}
N.~Alon, G.~Gutin, E.~J. Kim, S.~Szeider, A.~Yeo, Solving max-\emph{r}-sat
  above a tight lower bound, Algorithmica 61~(3) (2011) 638--655.
\newblock \href {http://dx.doi.org/10.1007/s00453-010-9428-7}
  {\path{doi:10.1007/s00453-010-9428-7}}.

\bibitem{Ausiello}
G.~Ausiello, P.~Crescenzi, V.~Gambosi, G.~Kann, A.~Marchetti-Spaccamela,
  M.~Protasi, Complexity and Approximation: Combinatorial optimization problems
  and their approximability properties, Springer-Verlag, 1999.

\bibitem{DBLP:journals/tcs/Kanj15}
I.~A. Kanj, G.~Lin, T.~Liu, W.~Tong, G.~Xia, J.~Xu, B.~Yang, F.~Zhang,
  P.~Zhang, B.~Zhu, \href{http://dx.doi.org/10.1016/j.tcs.2015.06.010}{Improved
  parameterized and exact algorithms for cut problems on trees}, Theor. Comput.
  Sci. 607 (2015) 455--470.
\newblock \href {http://dx.doi.org/10.1016/j.tcs.2015.06.010}
  {\path{doi:10.1016/j.tcs.2015.06.010}}.
\newline\urlprefix\url{http://dx.doi.org/10.1016/j.tcs.2015.06.010}

\bibitem{DBLP:journals/algorithmica/GargVY97}
N.~Garg, V.~V. Vazirani, M.~Yannakakis, Primal-dual approximation algorithms
  for integral flow and multicut in trees, Algorithmica 18~(1) (1997) 3--20.
\newblock \href {http://dx.doi.org/10.1007/BF02523685}
  {\path{doi:10.1007/BF02523685}}.

\bibitem{Chen2012}
J.~Chen, J.~Fan, I.~A. Kanj, Y.~Liu, F.~Zhang, Multicut in trees viewed through
  the eyes of vertex cover, J. Comput. Syst. Sci. 78~(5) (2012) 1637--1650.

\bibitem{Dinur2004}
I.~Dinur, S.~Safra, On the hardness of approximating minimum vertex cover,
  Annals of Mathematics 162.

\bibitem{Khot2008}
S.~Khot, O.~Regev, \href{https://doi.org/10.1016/j.jcss.2007.06.019}{Vertex
  cover might be hard to approximate to within 2-epsilon}, J. Comput. Syst.
  Sci. 74~(3) (2008) 335--349.
\newblock \href {http://dx.doi.org/10.1016/j.jcss.2007.06.019}
  {\path{doi:10.1016/j.jcss.2007.06.019}}.
\newline\urlprefix\url{https://doi.org/10.1016/j.jcss.2007.06.019}

\bibitem{LokshtanovMS11}
D.~Lokshtanov, D.~Marx, S.~Saurabh,
  \href{http://albcom.lsi.upc.edu/ojs/index.php/beatcs/article/view/96}{{Lower
  bounds based on the Exponential Time Hypothesis}}, Bulletin of the {EATCS}
  105 (2011) 41--72.
\newline\urlprefix\url{http://albcom.lsi.upc.edu/ojs/index.php/beatcs/article/view/96}

\bibitem{DBLP:conf/mfcs/KomusiewiczN12}
C.~Komusiewicz, R.~Niedermeier, New races in parameterized algorithmics, in:
  B.~Rovan, V.~Sassone, P.~Widmayer (Eds.), Mathematical Foundations of
  Computer Science 2012 - 37th International Symposium, {MFCS} 2012,
  Bratislava, Slovakia, August 27-31, 2012. Proceedings, Vol. 7464 of LNCS,
  Springer, 2012, pp. 19--30.
\newblock \href {http://dx.doi.org/10.1007/978-3-642-32589-2_2}
  {\path{doi:10.1007/978-3-642-32589-2_2}}.

\bibitem{DBLP:journals/jacm/AlonYZ95}
N.~Alon, R.~Yuster, U.~Zwick, Color-coding, J. {ACM} 42~(4) (1995) 844--856.
\newblock \href {http://dx.doi.org/10.1145/210332.210337}
  {\path{doi:10.1145/210332.210337}}.

\bibitem{DBLP:books/daglib/0023376}
T.~H. Cormen, C.~E. Leiserson, R.~L. Rivest, C.~Stein, Introduction to
  Algorithms {(3.} ed.), {MIT} Press, 2009.

\bibitem{DBLP:journals/jal/MahajanR99}
M.~Mahajan, V.~Raman, Parameterizing above guaranteed values: Maxsat and
  maxcut, J. Algorithms 31~(2) (1999) 335--354.
\newblock \href {http://dx.doi.org/10.1006/jagm.1998.0996}
  {\path{doi:10.1006/jagm.1998.0996}}.

\end{thebibliography}

\end{document}